\documentclass[journal,12pt,onecolumn,draftclsnofoot]{IEEEtran}

\usepackage{bm}
\usepackage{amsmath}
\usepackage{amsthm}
\usepackage{graphicx}
\usepackage{subfigure}
\usepackage{amsfonts}
\usepackage{amssymb}
\usepackage{algorithm}
\usepackage{algpseudocode}
\usepackage{epsfig}
\usepackage{cite}
\usepackage{threeparttable}
\usepackage{color}
\usepackage{soul}

\ifCLASSINFOpdf

\else

\fi

\begin{document}

\title{\LARGE{Receiver Design for MIMO Unsourced Random Access with SKP Coding}}

\author{Zeyu Han,~\IEEEmembership{Graduate Student Member,~IEEE,} \\
        Xiaojun Yuan,~\IEEEmembership{Senior Member,~IEEE,}
        Chongbin Xu,~\IEEEmembership{Member,~IEEE,} \\
        and Xin Wang,~\IEEEmembership{Senior Member,~IEEE}

\thanks{Zeyu Han, Chongbin Xu, and Xin Wang are with the Key Laboratory for Information Science of Electromagnetic Waves (MoE), Department of Communication Science and Engineering, Fudan University, Shanghai 200433, China (e-mail: chbinxu@fudan.edu.cn).}
\thanks{Xiaojun Yuan is with the National Key Laboratory of Science and Technology on Communications, University of Electronic Science and Technology of China, Chengdu 611731, China (e-mail: xjyuan@uestc.edu.cn).}
\vspace{-1cm}
}

\markboth{}
{Han \MakeLowercase{\textit{et al.}}: Receiver Design for MIMO Unsourced Random Access with SKP Coding}

\maketitle
\begin{abstract}
In this letter, we extend the sparse Kronecker-product (SKP) coding scheme, originally designed for the additive white Gaussian noise (AWGN) channel, to multiple input multiple output (MIMO) unsourced random access (URA). With the SKP coding adopted for MIMO transmission, we develop an efficient Bayesian iterative receiver design to solve the intended challenging trilinear factorization problem. Numerical results show that the proposed design outperforms the existing counterparts, and that it performs well in all simulated settings with various antenna sizes and active-user numbers.
\end{abstract}


\IEEEpeerreviewmaketitle

\clearpage

\vspace{-0.3cm}
\section{Introduction}

\IEEEPARstart{M}{assive} machine-type communication (mMTC) has become one of the most important scenarios in future wireless communications \cite{beyond,URA}. Unsourced random access (URA), proposed by Polyanskiy in \cite{Unsourced}, provides a new paradigm for mMTC that simplifies system scheduling significantly. The main idea of URA is that all users share a common codebook, and the access point (AP) only needs to recover the set of active codewords in this common codebook. Various works on theoretical analysis and practical coding design for URA have been proposed \cite{SKPC,sa7_8,sa_Rayleigh}.

\textcolor{black}{Recently, URA with multiple receiving antennas at AP, commonly referred to as multiple input multiple output (MIMO) URA, has attracted growing research interest due to its capability in supporting more users and improving spectrum efficiency \cite{MIMOURA,block3,block4,block5_slow,sl_LDPC, TBM, R1, R2}.} The first coding scheme for MIMO URA is proposed in \cite{MIMOURA} based on the segmented tree code. Various optimizations on the tree code and related receiver design are further investigated in \cite{block3,block4,block5_slow}, and noticeable performance improvement has been observed. \textcolor{black}{For the quasi-static channel case, two-phase approaches are proposed to estimate the channels first and then use them for multi-user detection \cite{sl_LDPC, R1}. In \cite{R2}, the fading spread unsourced random access (FASURA) scheme is developed to improve such two-phase approaches by using the data-aided noisy pilot channel estimation (NOPICE) and provides the state-of-the-art results.} 

\textcolor{black}{A tensor approach is proposed in \cite{TBM}.} Specifically, the data of each user is encoded and modulated as a rank-one tensor, and it is recovered by solving a canonical polyadic decomposition (CPD) problem at the receiver. However, it is hard to incorporate soft-in soft-out channel decoding techniques into this design due to the lack of efficient Bayesian CPD algorithms. In this case, only a limited coding gain is observed.

In this letter, we extend the sparse Kronecker-product (SKP) coding scheme in \cite{SKPC} from the AWGN channel to the MIMO URA scenario. The SKP coding shows a noticeable coding gain and a good multiple access capability in the AWGN channel. However, the receiver design for MIMO URA with SKP leads to a challenging trilinear problem. To solve the problem, we convert it to a bilinear problem with structured factor matrices, based on which an efficient Bayesian iterative receiver is then designed. Numerical results show that, the proposed design outperforms the existing schemes, and that it performs well in all simulated settings with various antenna sizes and active-user numbers.

\emph{Notation}: $\otimes$ denotes the Kronecker product; $\mathbb{C}$ denotes the complex number field; $\mathcal{CN}(\mu,\sigma^2)$ denotes the complex Gaussian distribution with mean $\mu$ and variance $\sigma^2$; ${{\mathcal{CN}}}(a; \mu,\sigma^2)$ denotes the probability density value of random variable $x\sim \mathcal{CN}(\mu,\sigma^2)$ at $x=a$; $\|\cdot\|_2$ denotes the $\ell^2$-norm; \textcolor{black}{and $\delta(\cdot)$ denotes Dirac delta function.}

\vspace{-0.3cm}
\section{System Model \& SKP Encoding}

\subsection{System Model}
We consider a MIMO URA system with an $M$-antenna AP and $K$ single-antenna users. Among the $K$ users, only $K_a$ of them are active in a transmission frame \cite{Unsourced}. Denote by $\bm{h}_j\in\mathbb{C}^{M\times1}$ the channel of active user $j$ seen by the AP. Each active user $j$ selects a codeword $\bm{v}_j\in\mathbb{C}^{T_{tot}\times1}$ (corresponding to a $B$-bit packet) from the common codebook $\mathbb{V}$ over $T_{tot}$ complex channel usages for transmission. Following \cite{block5_slow,sl_LDPC,TBM, R1, R2}, we assume that the channel vectors $\{ \bm{h}_j \}$ are constant over the $T_{tot}$ channel usages. Then the received signal $\bm{y}$ at AP is modeled as

\vspace{-0.3cm}
\begin{equation}
\bm{y}=\sum_{j=1}^{K_a}\bm{h}_j\otimes\bm{v}_j+\bm{w} \label{eori}
\end{equation}
where $\bm{w}\sim\mathcal{CN}(0, N_0 \bm{I}_{MT_{tot}})$ is the additive white Gaussian noise with the power density $N_0$. Each active user $j$ is subject to the power constraint $\textbf{E}\{\|\bm{v}_j\|_2^2\}\leq P$, and the average bit SNR is defined as $E_b/N_0\triangleq{P}/(B N_0)$.

The task of the receiver is to recover the transmitted codewords (i.e., packets) based on $\bm{y}$, i.e., to output the estimated list $\mathcal{L}(\bm{y})$ consisting of at most $K_a$ codewords. Following \cite{Unsourced}, the per-user probability of error (PUPE) is defined as $P_e\triangleq\frac{1}{K_a}\sum_{j=1}^{K_a}\Pr\big\{\mathbb{E}_j\big\}$, where $\mathbb{E}_j\triangleq\big\{\bm{v}_j\notin \mathcal{L}(\bm{y})\big\}\cup\big\{\exists i\neq j,\ \bm{v}_j=\bm{v}_i\big\}$ is the error event of active user $j$. {\color{black} Then the system performance can be measured by the lowest possible ${E_b}/{N_0}$ that meets the given PUPE requirement $P_e\leq \varepsilon$. }

\vspace{-0.3cm}
\subsection{SKP Coding}
We adopt the SKP coding scheme in \cite{SKPC}, where the common codebook $\mathbb{V}=\{\bm{v}=\bm{a}\otimes\bm{x}|\bm{a}\in\mathbb{V}_a, \bm{x}\in\mathbb{V}_x\}$ with $\bm{a}$ and $\bm{x}$ being sparse and forward-error-correction (FEC) codewords respectively. In particular, with the $B$-bit packet $\bm{b}=[\bm{b}^{(a)}; \bm{b}^{(x)}]$, we encode $B_a$-bit $\bm{b}^{(a)}$ and $B_x$-bit $\bm{b}^{(x)}$ into $\bm{a}$ and $\bm{x}$ with $B\!=\!B_a+B_x$ in the following way.

\begin{itemize}
  \item $\bm{b}^{(a)}$ is encoded by the index modulation (IM) as a length-$L_a$ vector $\bm{a}$ consisting of $I_{IM}$ length-$L_{IM}$ IM segments with $L_a=I_{IM}L_{IM}$. Each segment only contains one non-zero element taking its value from the signal constellation $\mathcal{S}$. For convenience, we refer to the non-zero positions of $\bm{a}$ as the support of $\bm{a}$.
  \item \textcolor{black}{ $\bm{b}^{(x)}$ is encoded by the FEC code and then modulated over $\mathcal{S}$, yielding a length-$L_x$ vector $\bm{x}$.}
\end{itemize}
Then the received signal $\bm{y}$ at AP can be written as

\vspace{-0.4cm}
\begin{equation}
\bm{y}=\sum_{j=1}^{K_a}\bm{h}_j\otimes\bm{v}_j+\bm{w}=\sum_{j=1}^{K_a} \bm{h}_j\otimes\bm{a}_j\otimes\bm{x}_j+\bm{w}.
\vspace{-0.15cm}
\label{aori1}
\end{equation}
\textcolor{black}{Notice that the reconstruction of $\{\bm{h}_j, \bm{a}_j, \bm{x}_j\}$ from $\bm{y}$ suffers from ambiguity. To eliminate ambiguity, we set the non-zero element in the first segment of $\bm{a}_j$ and the first $e_{Ref}$ elements of $\bm{x}_j$ as the reference symbols (known by the receiver).}

{\color{black} The focus of this paper is to design an efficient receiver to recover $\{ {\bm{h}}_j \}$, $\{ \bm{a}_j \}$ and $\{ \bm{x}_j \}$ from $\bm{y}$}, which is a challenging trilinear factorization problem, a.k.a. the CPD problem with three dimensions. Many existing algorithms, e.g., alternating least squares (ALS), Gauss-Newton method, line-search, etc., can be adopted to solve the problem\cite{CPD_survey}. However, these algorithms are non-Bayesian and cannot exploit the \emph{a priori} knowledge of $\{ {\bm{h}}_j \}$, $\{ \bm{a}_j \}$ and $\{ \bm{x}_j \}$. To improve the performance, we propose to tackle the problem based on Bayesian principles, as detailed in the next section.

\vspace{-0.3cm}
\section{Bayesian Receiver Design}

In this section, we design an efficient Bayesian receiver to recover the data from $\bm{y}$ based on (\ref{aori1}) as follows.

\vspace{-0.3cm}
\subsection{Problem Formulation}
\vspace{-0.1cm}

Considering different combinations of vectors $\bm{h}_j$, $\bm{a}_j$ and $\bm{x}_j$, we can rewrite (\ref{aori1}) as

\begin{equation}
\bm{y}=\sum_{j=1}^{K_a} (\bm{h}_j\otimes\bm{a}_j)\otimes\bm{x}_j+\bm{w}, \label{aori}
\end{equation}

\begin{equation}
\bm{y}=\sum_{j=1}^{K_a} \bm{h}_j\otimes(\bm{a}_j\otimes\bm{x}_j)+\bm{w}, \label{aori_2}
\end{equation}
or alternatively

\begin{equation}
\tilde{\bm{y}}=\sum_{j=1}^{K_a} (\bm{h}_j\otimes\bm{x}_j)\otimes\bm{a}_j+\tilde{\bm{w}} \label{aori_3}
\end{equation}
where $\tilde{\bm{y}}$ is a certain permutation of $\bm{y}$. The receiver designs based on (\ref{aori}), (\ref{aori_2}) and (\ref{aori_3}) are similar. Due to space limitation, we focus on the design with (\ref{aori}) in the following.

Denote $\bm{G} = [\bm{h}_1 \otimes \bm{a}_1, \cdots, \bm{h}_{K_a} \otimes \bm{a}_{K_a}] = [\bm{g}_1, \cdots, \bm{g}_{K_a}] \in \mathbb{C}^{M\!L_a\times K_a}$, $\bm{X} = {[\bm{x}_1, \cdots, \bm{x}_{K_a}]}^T\in \mathbb{C}^{K_a\times L_x}$, and reshape $\bm{y}$ to $\bm{Y}\in \mathbb{C}^{M\!L_a\times L_x}$, and $\bm{w}$ to $\bm{W}$ accordingly. Then (\ref{aori}) can be rewritten in a matrix form as
\vspace{-0.10cm}
\begin{equation} \label{dec_mat}
\bm{Y}=\sum_{j=1}^{K_a} (\bm{h}_j\otimes\bm{a}_j)\bm{x}_j^T+\bm{W} = \sum_{j=1}^{K_a} \bm{g}_j \bm{x}_j^T + \bm{W} = \bm{G}\bm{X} + \bm{W}.
\end{equation}

Denote $\bm{H} = [\bm{h}_1, \cdots, \bm{h}_{K_a}]$ and $\bm{A} = [\bm{a}_1, \cdots, \bm{a}_{K_a}]$. The probability model corresponding to (\ref{dec_mat}) can be written as

\vspace{-0.4cm}
\begin{equation} \label{dec_prob}
\begin{split}
 p\big(\bm{G}, \bm{H}, \bm{A}, \bm{X} | \bm{Y}\big) 
& \propto p\big(\bm{Y} | \bm{G}, \bm{X}\big)  p\big(\bm{G} | \bm{H}, \bm{A}\big)  p\big(\bm{H}\big) p\big(\bm{A}\big)  p\big(\bm{X}\big) \\
& \propto e^{ \frac{\|\bm{Y}-\bm{GX}\|_F^2}{N_0} } \cdot \prod_{j}^{} \delta(\bm{g}_j\!-\!\bm{h}_j\!\otimes\!\bm{a}_j) p\big(\bm{h}_j\big) p\big(\bm{a}_j\big)  \prod_{j}^{} p\big(\bm{x}_j\big)
\end{split}
\end{equation}
where $p\big(\bm{h}_j\big)$ is the \textit{a priori} distribution of $\bm{h}_j$; $p\big(\bm{a}_j\big)$ and $p\big(\bm{x}_j\big)$ are given by the coding constraints of $\bm{a}_j$ and $\bm{x}_j$, respectively, i.e., $p\big(\bm{a}_j\big)=1/|\mathbb{V}_a|$ for $\bm{a}_j\in\mathbb{V}_a$ and $p\big(\bm{x}_j\big)=1/|\mathbb{V}_x|$ for $\bm{x}_j\in\mathbb{V}_x$. 
As $\bm{a}_j$ is IM-encoded, $\bm{G}$ is sparse. Thus, the recovery of $\bm{G}$ and $\bm{X}$ from $\bm{Y}$ in (\ref{dec_mat}) is a sparse matrix factorization problem. To solve this problem, we develop an efficient Bayesian receiver, as described in the following.

\vspace{-0.3cm}
\subsection{Bayesian Message-Passing Algorithm}

Fig. \ref{receiver} shows the factor graph of the probability model in (\ref{dec_prob}), consisting of two types of nodes below:
\begin{itemize}
  \item Variable nodes $\bm{Y}$, $\bm{G}$, $\{\bm{g}_j\}$, $\{\bm{h}_j\}$, $\{\bm{a}_j\}$, $\bm{X}$, and $\{\bm{x}_j\}$ depicted as white circles;
  \item Check nodes $p(\bm{Y} | \bm{G}, \bm{X})$, $\{\delta(\bm{g}_j-\bm{h}_j\otimes\bm{a}_j)\}$, $\{p(\bm{h}_j)\}$, $\{p(\bm{a}_j)\}$, and $\{p(\bm{x}_j)\}$ depicted as black boxes.
\end{itemize}

\vspace{-0.2cm}
\begin{figure}[ht]
\centering
\includegraphics[width=0.56\textwidth]{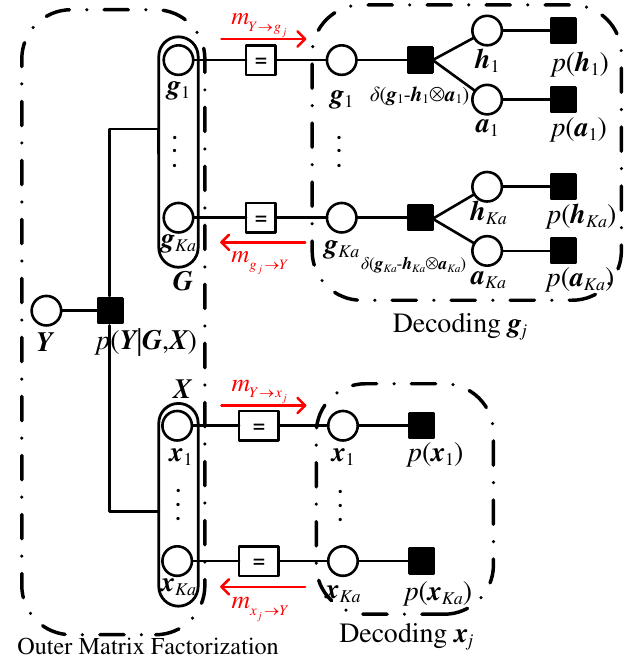} \\
\vspace{-0.2cm}
\caption{The factor graph of the proposed Bayesian algorithm, where ``$=$'' denotes the equality constraint.} 
\label{receiver}
\end{figure}

Based on the factor graph, we divide the whole receiver into three modules: Outer matrix factorization to estimate $\bm{G}$ and $\bm{X}$ from $\bm{Y}$ by ignoring the coding structures of $\bm{G}$ and $\bm{X}$; decoding of $\{\bm{g}_j\}$; and decoding of $\{\bm{x}_j\}$. Denote by $\{m_{\bm{Y}\rightarrow\bm{g}_j}\}$\ and $\{m_{\bm{g}_j\rightarrow\bm{Y}}\}$ the messages passed between the factorization and decoding-$\bm{g}_j$ module; and $\{m_{\bm{Y}\rightarrow\bm{x}_j}\}$ and $\{m_{\bm{x}_j\rightarrow\bm{Y}}\}$ the messages passed between the factorization and decoding-$\bm{x}_j$ module. The detailed operations of the above three modules are described as follows.

\subsubsection{Outer Matrix Factorization}
This module estimates $\bm{G}$ and $\bm{X}$ based on the received signal $\bm{Y}$ and the feedback messages of $\bm{G}$ and $\bm{X}$ from the other modules given in the element-wise forms, i.e., $m_{\bm{g}_j\rightarrow\bm{Y}}(\bm{g}_j)=\prod_{l} m_{g_{j,l}\rightarrow\bm{Y}}(g_{j,l})=\prod_{l} \big[(1-\lambda_{j,l})\delta(g_{j,l}) + \lambda_{j,l}\mathcal{CN}(g_{j,l}; \mu_{j,l},\tau_{j,l})\big]$ and $m_{\bm{x}_j\rightarrow\bm{Y}}(\bm{x}_j)=\prod_{l_x} m_{x_{j,l_x}\rightarrow\bm{Y}}(x_{j,l_x})=\prod_{l_x} \big[\sum_{i=1}^{|\mathcal{S}|} \Phi_{j,l_x}(s_i)\delta(x_{j,l_x}-s_i)\big]$ with \textcolor{black}{$s_i\in\mathcal{S}$ (the signal constellation), where $\{\lambda_{j,l}, \mu_{j,l}, \tau_{j,l}, \Phi_{j,l_x}(s_i)\}$ are all scalar parameters in the messages}, and are initialized based on the \emph{a priori} knowledge of $\{ \bm{g}_j \}$ and $\{ \bm{x}_j \}$.

We adopt the bilinear generalized approximate message passing (BiG-AMP) algorithm \cite{BiG1}, \cite{SSL}, \cite{SKPC} for the outer matrix factorization. BiG-AMP is a Bayesian sparse matrix factorization algorithm based on the sum-product rule and Gaussian message approximation. \textcolor{black}{BiG-AMP consists of four main steps. Denote by $\bm{Z}=\bm{G}\bm{X}$. Firstly, the messages of $\bm{Z}$ are calculated based on the messages of $\bm{G}$ and $\bm{X}$. Secondly, by further combining $\bm{Y}$, the \textit{a posteriori} means and variances of $\bm{Z}$ are obtained. Thirdly, the messages of $\bm{G}$ (or $\bm{X}$) are calculated by using the messages of $\bm{Z}$ and $\bm{X}$ (or $\bm{G}$) according to $\bm{Z}=\bm{G}\bm{X}$. Finally, the \textit{a posteriori} means and variances of $\bm{G}$ (or $\bm{X}$) are obtained by further combining the \textit{a priori} information of $\bm{G}$ (or $\bm{X}$). The four steps iterate to obtain the estimates of $\bm{G}$ and $\bm{X}$. The detailed description can be found in TABLE \ref{e1}. }

\begin{table}[H]
\caption{The BiG-AMP algorithm}
\label{e1}
\begin{equation*}
\begin{array}{|lrcl@{}r|}\hline
  \multicolumn{2}{|l}{\textsf{definitions:}}&&&\\
  &p_{{}{z}_{ml}|{}{p}_{ml}}(z|\hat{p};v_p)
   &\triangleq& \frac{p_{{}{y}_{ml}|{}{z}_{ml}}(y_{ml}|z) \,\mathcal{CN}(z;\hat{p},v_p)}
	{\int_{z'} p_{{}{y}_{ml}|{}{z}_{ml}}(y_{ml}|z') \,\mathcal{CN}(z';\hat{p},v_p)} &\\
  &p_{{}{x}_{nl}}(x|\hat{x};v_x)
   &\triangleq& \frac{p_{{}{x}_{nl}}\!(x) \,\mathcal{CN}(x;\hat{x},v_x)}
        {\int_{x'}p_{{}{x}_{nl}}\!(x') \,\mathcal{CN}(x';\hat{x},v_x)}&\\
  &p_{{}{g}_{mn}}(g|\hat{g};v_g)
    &\triangleq& \frac{p_{{}{g}_{mn}}\!(g) \,\mathcal{CN}(g;\hat{g},v_g)}
         {\int_{g'}p_{{}{g}_{mn}}\!(g') \,\mathcal{CN}(g';\hat{g},v_g)}&\\
  \multicolumn{2}{|l}{\textsf{initialization:}}&&&\\
  &\forall m,l:
   \hat{s}_{ml}(0) &=& 0 &\\
  &\forall m,n,l: \textsf{choose~} &
  \multicolumn{2}{l}{\tilde{x}_{nl}(1), v_{\tilde{x}_{nl}}(1), \tilde{g}_{mn}(1), v_{\tilde{g}_{mn}}(1)} &\\
  \multicolumn{2}{|l}{\textsf{for $t=1,\dots T_\textrm{max}$}}&&&\\
  &\forall m,l:
   v_{\bar{p}_{ml}}(t)
   &=& \textstyle \sum_{n=1}^{N} \left[|\tilde{g}_{mn}(t)|^2 v_{\tilde{x}_{nl}}(t) + v_{\tilde{g}_{mn}}(t) |\tilde{x}_{nl}(t)|^2\right] \!\!& \text{(A1)}\\
  &\forall m,l:
   \bar{p}_{ml}(t)
   &=& \textstyle \sum_{n=1}^{N} \tilde{g}_{mn}(t) \tilde{x}_{nl}(t) & \text{(A2)}\\
    &\forall m,l:
   v_{p_{ml}}(t)
   &=& \textstyle v_{\bar{p}_{ml}}(t) + \sum_{n=1}^{N} v_{\tilde{g}_{mn}}(t) v_{\tilde{x}_{nl}}(t) & \text{(A3)}\\
  &\forall m,l:
   \hat{p}_{ml}(t) &=&
   \textstyle \bar{p}_{ml}(t) - \hat{s}_{ml}(t\!-\!1)v_{\bar{p}_{ml}}(t)& \text{(A4)}\\
  &\forall m,l:
   v_{\tilde{z}_{ml}}(t) &=&
   \textstyle \textrm{var}\{z_{ml}\sim p_{z_{ml}|p_{ml}}(z_{ml}|\hat{p}_{ml}(t);v_{p_{ml}}(t))\} & \text{(A5)}\\
  &\forall m,l:
    \tilde{z}_{ml}(t) &=&
   \textstyle  \textrm{E}\{z_{ml}\sim p_{z_{ml}|p_{ml}}(z_{ml}|\hat{p}_{ml}(t);v_{p_{ml}}(t))\} & \text{(A6)}\\
  &\forall m,l:
   v_{s_{ml}}(t) &=&
   \textstyle {(1 -  v_{\tilde{z}_{ml}}(t)/v_{p_{ml}}(t))/v_{p_{ml}}(t)}  & \text{(A7)}\\
  &\forall m,l:
   \hat{s}_{ml}(t) &=&
   	\textstyle ( \tilde{z}_{ml}(t) - \hat{p}_{ml}(t))/v_{p_{ml}}(t) & \text{(A8)}\\
 &\forall n,l:
   v_{x_{nl}}(t)
   &=& \textstyle \big(\sum_{m=1}^{M} |\tilde{g}_{mn}(t)|^2 v_{s_{ml}}(t)
	\big)^{-1} & \text{(A9)}\\
  &\forall n,l:
   \hat{x}_{nl}(t)
   &=& \textstyle \tilde{x}_{nl}(t) ( 1 - v_{x_{nl}}(t) \sum_{m=1}^{M} v_{\tilde{g}_{mn}}(t) v_{s_{ml}}(t)   ) &\\
   &&&\qquad+ v_{x_{nl}}(t) \sum_{m=1}^{M} \tilde{g}^{*}_{mn}(t)
	\hat{s}_{ml}(t)  & \text{(A10)}\\
  &\forall m,n:
   v_{g_{mn}}(t)
   &=& \textstyle \big(\sum_{l=1}^{L} |\tilde{x}_{nl}(t)|^2 v_{s_{ml}}(t)
	\big)^{-1} & \text{(A11)}\\
  &\forall m,n:
   \hat{g}_{mn}(t)
   &=& \textstyle \tilde{g}_{mn}(t)(1 - v_{g_{mn}}(t) \sum_{l=1}^{L} v_{\tilde{x}_{nl}}(t) v_{s_{ml}}(t)    ) &\\
   &&&\qquad+ v_{g_{mn}}(t) \sum_{l=1}^{L} \tilde{x}^{*}_{nl}(t)
	\hat{s}_{ml}(t)  & \text{(A12)}\\
  &\forall n,l:
	v_{\tilde{x}_{nl}}(t\!+\!1) &=&
		 \textrm{var}\{x_{nl}\sim p_{x_{nl}}(x_{nl}|\hat{x}_{nl}(t); v_{x_{nl}}(t))\} & \text{(A13)}\\
  &\forall n,l:
    \tilde{x}_{nl}(t\!+\!1) &=&
    	 \textrm{E}\{x_{nl}\sim p_{x_{nl}}(x_{nl}|\hat{x}_{nl}(t); v_{x_{nl}}(t))\} & \text{(A14)}\\
  &\forall m,n:
	v_{\tilde{g}_{mn}}(t\!+\!1) &=&
		 \textrm{var}\{g_{mn}\sim p_{g_{mn}}(g_{mn}|\hat{g}_{mn}(t); v_{g_{mn}}(t))\}& \text{(A15)}\\
  &\forall m,n:
	\tilde{g}_{mn}(t\!+\!1) &=&
		 \textrm{E}\{g_{mn}\sim p_{g_{mn}}(g_{mn}|\hat{g}_{mn}(t); v_{g_{mn}}(t))\} & \text{(A16)}\\
   \multicolumn{4}{|c}{\textsf{if $\sum_{m,l} |\bar{p}_{ml}(t) - \bar{p}_{ml}(t\!-\!1)|^2 \le \tau_\textrm{BiG-AMP} \sum_{m,l} |\bar{p}_{ml}(t)|^2$, {\textsf{stop}}}}&\\
    \multicolumn{2}{|l}{\textsf{end}}&&&\\\hline
\end{array}
\end{equation*}

\textit{Hint: For clarity, the notations used here are slightly different from those in the reference \cite{BiG1}.}
\end{table}

Following the message passing principle \cite{BP}, the outputs of this module to the other two decoding modules are the extrinsic messages of $\bm{G}$ and $\bm{X}$\footnote{{The phase ambiguity of $\bm{G}$ and $\bm{X}$ is eliminated based on the reference symbols in $\bm{X}$.}} \textcolor{black}{(obtained by excluding the input \textit{a priori} information of $\bm{G}$ and $\bm{X}$ from the output of BiG-AMP)} according to (A9)--(A12), denoted by $\{m_{\bm{Y}\rightarrow\bm{g}_j}\}$ and $\{m_{\bm{Y}\rightarrow\bm{x}_j}\}$ respectively. With these outputs, the estimates of $\bm{G}$ and $\bm{X}$ can be further refined in the two decoding modules, as detailed below.

\subsubsection{Decoding $\bm{g}_j$}
This module estimates $\bm{g}_j=\bm{h}_j\otimes\bm{a}_j$ based on the message $m_{\bm{Y}\rightarrow\bm{g}_j}(\bm{g}_j)=\prod_{l}{m_{\bm{Y}\rightarrow{g_{j,l}}}}(g_{j,l})$ and the prior information $p\big(\bm{h}_j\big)$ and $p\big(\bm{a}_j\big)$. From the sum-product rule, each $m_{g_{j,l}\rightarrow\bm{Y}}(g_{j,l})$ can be calculated as

\color{black}
\vspace{-0.3cm}
\begin{equation} \label{all_msg}
\begin{split}
m_{g_{j,l}\rightarrow\bm{Y}}(g_{j,l})
= \sum_{\bm{a}_j\in \mathbb{V}_a} \int_{\bm{g}_j\backslash \{g_{j,l}\}, \bm{h}_j} \delta(\bm{g}_j-\bm{h}_j\otimes \bm{a}_j) \cdot \Big(\prod_{l'\neq l}m_{\bm{Y}\rightarrow {g_{j,l'}}}(g_{j,l'})\Big) p(\bm{h}_j) p(\bm{a}_j).
\end{split}
\end{equation}
\color{black}
With $\bm{a}^{(n)}$ denoting the $n$-th codeword in $\mathbb{V}_a$, the exact calculation of (\ref{all_msg}) needs to enumerate all $\bm{a}_j\in \mathbb{V}_a=[\bm{a}^{(1)}, \cdots, \bm{a}^{(2^{B_a})}]$, which is computationally prohibitive. Here we propose a low-complexity approximate method. We select a subset $\tilde{\mathbb{V}}_a=[\bm{a}^{(i_1)}, \cdots, \bm{a}^{(i_{N_{top}})}]\subset \mathbb{V}_a$ including $N_{top}$ choices of $\bm{a}_j$ corresponding to the $N_{top}$ terms in (\ref{all_msg}); and calculate (\ref{all_msg}) approximately based on $\tilde{\mathbb{V}}_a$. This method is similar to the list decoding method. Recall that $\bm{g}_j=\bm{h}_j\otimes\bm{a}_j$. The optimal $\tilde{\mathbb{V}}_a$ needs to consider all possible realizations of $\bm{h}_j$, which is computationally prohibitive as well. In this paper, we obtain an initial $\tilde{\mathbb{V}}_a$ based on a simplified signal model of $\bm{g}_j$, and then refine $\tilde{\mathbb{V}}_a$ by using alternating optimization.

\color{black}

{\textit{Initialize $\tilde{\mathbb{V}}_a$}:
We propose a two-step method to initialize $\tilde{\mathbb{V}}_a$. We first select $N_{top}$ possible supports of $\bm{a}_j$, and then obtain $N_{top}$ choices of $\bm{a}_j$ based on the selected supports. The detailed exposition is provided in Appendix A.

{\textit{Refine $\tilde{\mathbb{V}}_a$}}:
With $\tilde{\mathbb{V}}_a$, we refine each candidate estimate $\hat{\bm{a}}_j=\bm{a}^{(i_{n})}$ of $\bm{a}_j$ by estimating $\bm{h}_j$ and $\bm{a}_j$ alternatingly. With $m_{\bm{Y}\rightarrow{\bm{g}_{j}}}(\bm{g}_j)=\mathcal{CN}(\bm{g}_j; \hat{\bm{g}}_{j}, \nu_{\bm{g}_{j}})$, we have

\vspace{-0.3cm}
\begin{equation} \label{msg_modeled}
\hat{\bm{g}}_{j} = \bm{h}_j \otimes \bm{a}_j + \sqrt{\nu_{\bm{g}_{j}}}\zeta_{\bm{g}_j}
\end{equation}
where $\zeta_{\bm{g}_j}\sim\mathcal{CN}(0,\bm{I}_{M\!L_a})$. Then,

\begin{itemize}
  \item For a fixed $\bm{a}_j=\hat{\bm{a}}_j$, the likelihood estimate of $\bm{h}_j$ is obtained by using the least square estimation. Then this estimate is combined with $p\big(\bm{h}_j\big)$ to obtain the \textit{a posteriori} estimate. Then the maximum \textit{a posteriori} (MAP) estimate $\hat{\bm{h}}_j$ is obtained consequently, which is used in the next step.
  \item For a fixed $\bm{h}_j=\hat{\bm{h}}_j$, the \textit{a posteriori} estimate of $\bm{a}_j$ can be obtained similarly, and so does the MAP estimate $\hat{\bm{a}}_j$, which is used in the next iteration.
\end{itemize}
The above two operations are iterated until convergence, and the final $\bm{a}_j=\hat{\bm{a}}_j$ is used to update $\bm{a}^{(i_n)}$ in $\tilde{\mathbb{V}}_a$. With all $\{\bm{a}^{(i_n)}\}_{n=1}^{N_{top}}$ updated, $\tilde{\mathbb{V}}_a$ is refined.

With $\tilde{\mathbb{V}}_a$ available, the output of this module can be calculated by replacing $\mathbb{V}_a$ with $\tilde{\mathbb{V}}_a$ in (\ref{all_msg}).


\subsubsection{Decoding $\bm{x}_j$}
This module estimates $\bm{x}_j$ based on the message $m_{\bm{Y}\rightarrow\bm{x}_j}$ and the coding constraint $p(\bm{x}_j)$.

\textcolor{black} {First, the $e_{Re\!f}$ reference symbols are decoded.} Following \cite{RIGM}, to facilitate the design of the iterative algorithm, the specific values $s_p$ of these reference signals are assumed to be unknown during the iteration and the knowledge of $x_{j,l_x} = s_p$ is used only in the final step to remove ambiguity. Since all the $e_{Re\!f}$ reference symbols take the same value, the messages of $\{ m_{x_{j,l'_x}\rightarrow\bm{Y}} \}_{l'_x=1}^{e_{Re\!f}}$ are calculated as

\begin{equation}
m_{x_{j,l'_x}\rightarrow\bm{Y}}(s_i) \propto \prod_{l'_x=1, l'_x\neq l_x}^{e_{Re\!f}} m_{\bm{Y}\rightarrow{x_{j,l'_x}}}(s_i),\qquad\forall s_i \in \mathcal{S}.
\label{ecc12}
\end{equation}

Then, the Bahl-Cocke-Jelinek-Raviv (BCJR) algorithm \cite{BCJR} is used for the decoding of the remaining $\{x_{j,l_x}\}_{l_x=e_{Re\!f}+1}^{L_x}$. The extrinsic information $m_{x_{j,l_x}\rightarrow\bm{Y}}(s_i)$ for each symbol in $\bm{x}_j$ can be calculated based on the log-likelihood ratio (LLR) of each encoded bit output by the BCJR algorithm. Consequently, $m_{\bm{x}_j\rightarrow\bm{Y}}$ is obtained.
\color{black}

The above three modules are iterated until convergence. Finally, the hard-decision results of $\bm{v}_j=\bm{a}_j\otimes\bm{x}_j$ ($\forall j$) can be obtained.

\vspace{-0.3cm}
\subsection{Packet Decision}
The iterative receiver may suffer from the bad initial points due to its suboptimality. To get a better final output, we run the Bayesian algorithm multiple times (with random initializations), and propose the following packet decision method.

This method maintains a pending list $\mathcal{Q}(\bm{Y})$ of the estimated packets, where the $p$-th packet in this list has two attributes: the codeword $\bm{v}_p$ and the priority $\mathbb{P}_p$. At the end of each trial, $K_a$ estimated packets are obtained. For each estimated packet $j$, if $\exists p\in\mathcal{Q}(\bm{Y}), \bm{v}_j = \bm{v}_p$, its priority increases by 1, i.e., $\mathbb{P}_p\gets\mathbb{P}_p+1$; otherwise the estimated packet is added into $\mathcal{Q}(\bm{Y})$ with its initial priority $\mathbb{P}_p\gets1$.

After $T_{max}$ trials or the number of high-priority packets $\big|\big\{p\in\mathcal{Q}(\bm{Y})\big|\mathbb{P}_p\geq \mathbb{P}_{thr}\big\}\big|\!\geq\!K_a$, the $K_a$ packets with the highest priorities are selected as the resulting list $\mathcal{L}(\bm{Y})$.

The overall Bayesian decoding algorithm is summarized in Algorithm \ref{odecode}. The complexity of this algorithm is $\mathcal{O}(MK_aT_{tot})$ (dominated by the outer matrix factorization). 

\begin{algorithm}[ht]
\begin{algorithmic}[1]
\Require
$\bm{Y}, M, N_0, K_a$, and SKP coding parameters.
\Ensure
Resulting list $\mathcal{L}(\bm{Y})$ ($|\mathcal{L}(\bm{Y})| \leq K_a$).
\State Initialize $\mathcal{Q}(\bm{Y})\gets\varnothing$.
\For{$t\gets1$ to $T_{max}$}
\State Run the Bayesian iterative algorithm to obtain the estimates of $\bm{v}_j$ ($\forall j$).
\State Use packet decision method based on the above estimates to update $\mathcal{Q}(\bm{Y})$.
\State \textbf{If} $\big|\big\{p\in\mathcal{Q}(\bm{Y})\big|\mathbb{P}_p\geq \mathbb{P}_{thr}\big\}\big|\geq K_a$, \textbf{break}
\EndFor
\State Let $\mathcal{L}(\bm{Y})$ be the $K_a$ packets with the highest $\mathbb{P}_p$ in $\mathcal{Q}(\bm{Y})$.
\end{algorithmic}
\caption{Overall Bayesian decoding algorithm}
\label{odecode}
\end{algorithm}

\vspace{-0.3cm}
\section{Numerical Results}

\subsection{ \textcolor{black}{Performance Limit}}
The \textcolor{black}{capacity limit} of URA is generally difficult to evaluate. The random coding bounds for single-antenna systems are derived in \cite{Unsourced} for the AWGN channel and in \cite{sa_Rayleigh} for the fading channel. However, to the best of our knowledge, \textcolor{black}{no results on capacity limit} for MIMO URA have been reported. This is probably due to the difficulty in finding the optimal decoding subset of users in the case of non-degraded channels.

\textcolor{black}{To provide a benchmark for MIMO URA, we derive a limit of its PUPE by extending the channel $\bm{h}_j$ of each user $j$ as a diagonal matrix $\tilde{\bm{H}_j}\triangleq \|\bm{h}_j\|_2 \cdot \bm{I}\in \mathbb{C}^{M\times M}$ by following the idea in \cite{MEB} as}

\begin{equation} \label{expanded_01}
\begin{split}
\tilde{\bm{Y}}=\sum_{j=1}^{K_a} \bm{h_j}\bm{v_j}^T + \tilde{\bm{W}} = \sum_{j=1}^{K_a} \big(\|\bm{h}_j\|_2\bm{I}_M\big)\big(\|\bm{h}_j\|_2^{-1}\bm{h_j}\bm{v_j}^T\big) + \tilde{\bm{W}} = \sum_{j=1}^{K_a} \tilde{\bm{H}_j}\tilde{\bm{V}_j} + \tilde{\bm{W}}
\end{split}
\end{equation}
where $\tilde{\bm{V}_j}\triangleq\|\bm{h}_j\|_2^{-1}\bm{h_j}\bm{v_j}^T\in \mathbb{C}^{M\times T_{tot}}$. Here we ignore the correlation of each symbol in $\tilde{\bm{V}_j}$ and only remain the power constraint $\textbf{E}\{\|\tilde{v}_{j(\cdot)}\|_2^2\}\leq \frac{P}{M}$ ($\forall \tilde{v}_{j(\cdot)}\in \tilde{\bm{V}}_j$).

With such an extension, the optimal decoding subset of users in each channel realization can be readily determined. We then generalize the outage analysis for the single-antenna scenario in \cite{CapB} to the multiple-antenna scenario.

Specifically, assume that all $\bm{h}_j$ are known and the decoder ranks all users in descending order of $\|\bm{h}_j\|_2$, i.e., obtains a permutation $\{i_1,i_2,\cdots,i_{K_a}\}$ of $\{1,2,\cdots,K_a\}$ that satisfies $\|\bm{h}_{i_1}\|_2 \geq\|\bm{h}_{i_2}\|_2 \geq \cdots \geq\|\bm{h}_{i_{K_a}}\|_2$. Obviously, the joint decoder can decode $\tilde{K}_a$ users in the optimal subset $\{i_1, i_2, \cdots, i_{\tilde{K}_a}\}$ of $K_a$ successfully, when their rate vector is in the joint capacity region (with the remaining $(K_a-\tilde{K}_a)$ users as interference), i.e.,

\begin{equation} \label{expanded_02}
\begin{split}
T_{tot}\log_{2}{\left(1+\frac{\sum_{k'=\tilde{K}_a-k+1}^{\tilde{K}_a}{\|\bm{h}_{i_{k'}}\|_2^2 \frac{P}{M}}}{N_0+\sum_{k''=\tilde{K}_a+1}^{K_a}{\|\bm{h}_{i_{k''}}\|_2^2 \frac{P}{M}}}\right)} > B\cdot k,\qquad\forall k\in\{1,2,\cdots,\tilde{K}_a\}.
\end{split}
\end{equation}

\color{black}

We can obtain the maximum feasible $\tilde{K}_a$ for each channel realization by a line-search; we then obtain the limit of PUPE by averaging over all channel realizations. In Fig. \ref{result_all}, we calculate this PUPE limit for each $E_b/N_0$ value by averaging over $10^5$ channel realizations, based on which the $E_b/N_0$ bound to achieve a given PUPE is obtained.

\vspace{-0.3cm}
\subsection{Performance Comparison}

Here we use the settings in \cite{TBM} to compare the performance of the proposed scheme with the existing schemes \cite{MIMOURA,TBM,R1,R2}. Assume independent Rayleigh-fading channels (i.e., $p(\bm{h})\sim\mathcal{CN}(0,I_{M})$), $1\leq M \leq 50$, $10\leq K_a \leq 1500$, $T_{tot} = 3200$, $B = 96$ and $\varepsilon = 0.1$. \textcolor{black}{For SKP coding, the component FEC code is selected to be the tail-biting convolutional code (CC) due to its relatively good performance in short codeword lengths and its efficient BCJR algorithm for SISO decoding. The generator polynomial of CC is $[23\ 33]$ in octal form for low complexity with coding rate adjusted through puncturing. The modulation is set to be $\mathcal{S} = \{\pm \frac{\sqrt{2}}{2}\pm \frac{\sqrt{2}}{2}i\}$, i.e., $\pi/4$-quadrature phase shift keying (QPSK) modulation. The parameters of SKP coding are summarized in Table \ref{params}. (Here we select these parameters empirically. Further performance improvement is possible by optimizing these parameters.) In addition, $N_{top}=10$, $\mathbb{P}_{thr}=3$, and $T_{max}=30$.}

\begin{table}[ht]
\centering
\caption{parameters of SKP coding}
\label{params}
\resizebox{0.66\textwidth}{!}{
\begin{threeparttable}
	\begin{tabular}{|c||c|c|c|c|c|c|c|}
		\hline
        $\{M,K_a\}$ & Option & $L_{IM}$ & $L_a$ & $L_x$ & $e_{Re\!f}$ & $\bm{x}_j$\\ 
        \hline
        \{1,[10,40]\}; \{8,$\forall$\} & Eq. (\ref{aori}) & 8   & 40 & 80 & 7 & $R=1/2$ \textrm{CC}\\ 
        \hline
        \{1,[50,80]\}; \{8,$\forall$\} & Eq. (\ref{aori}) & 14  & 56 & 57 & 5 & $R=3/4$ \textrm{CC}\\ 
        \hline
        \{1,[85,95)\}; \{8,$\forall$\} & Eq. (\ref{aori}) & 26  & 78 & 41 & 1 & \textrm{Uncoded}\tnote{1}\\ 
        \hline
        \{50,$\forall$\} & Eq. (\ref{aori_2}) & 320  & 3200 & 1 & 1 & $B_x=0$\\ 
		\hline
	\end{tabular}
\begin{tablenotes}
\item[1] For $\bm{x}_j$, uncoded is the extreme case of the CC.
\end{tablenotes}
\end{threeparttable}}
\end{table}

\begin{figure*}[ht]
\centering
\subfigure[$M=1$]{
\begin{minipage}[c]{0.32\linewidth}
\centering
\includegraphics[width=\linewidth]{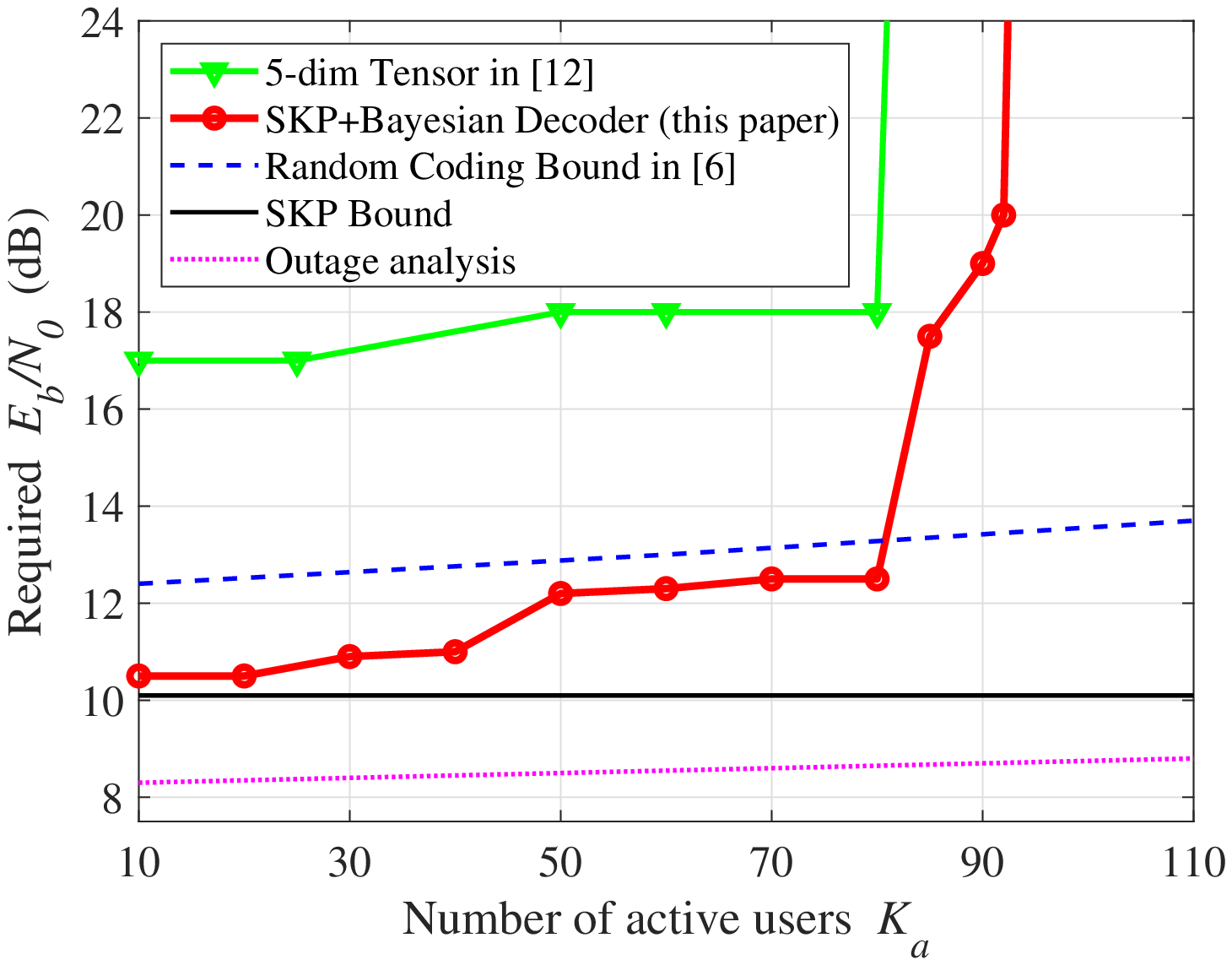}
\end{minipage}
}\hspace{-10pt}
\subfigure[\textcolor{black}{$M=50$ ($B = 100$, $\varepsilon = 0.05$)}]{
\begin{minipage}[c]{0.32\linewidth}
\centering
\includegraphics[width=\linewidth]{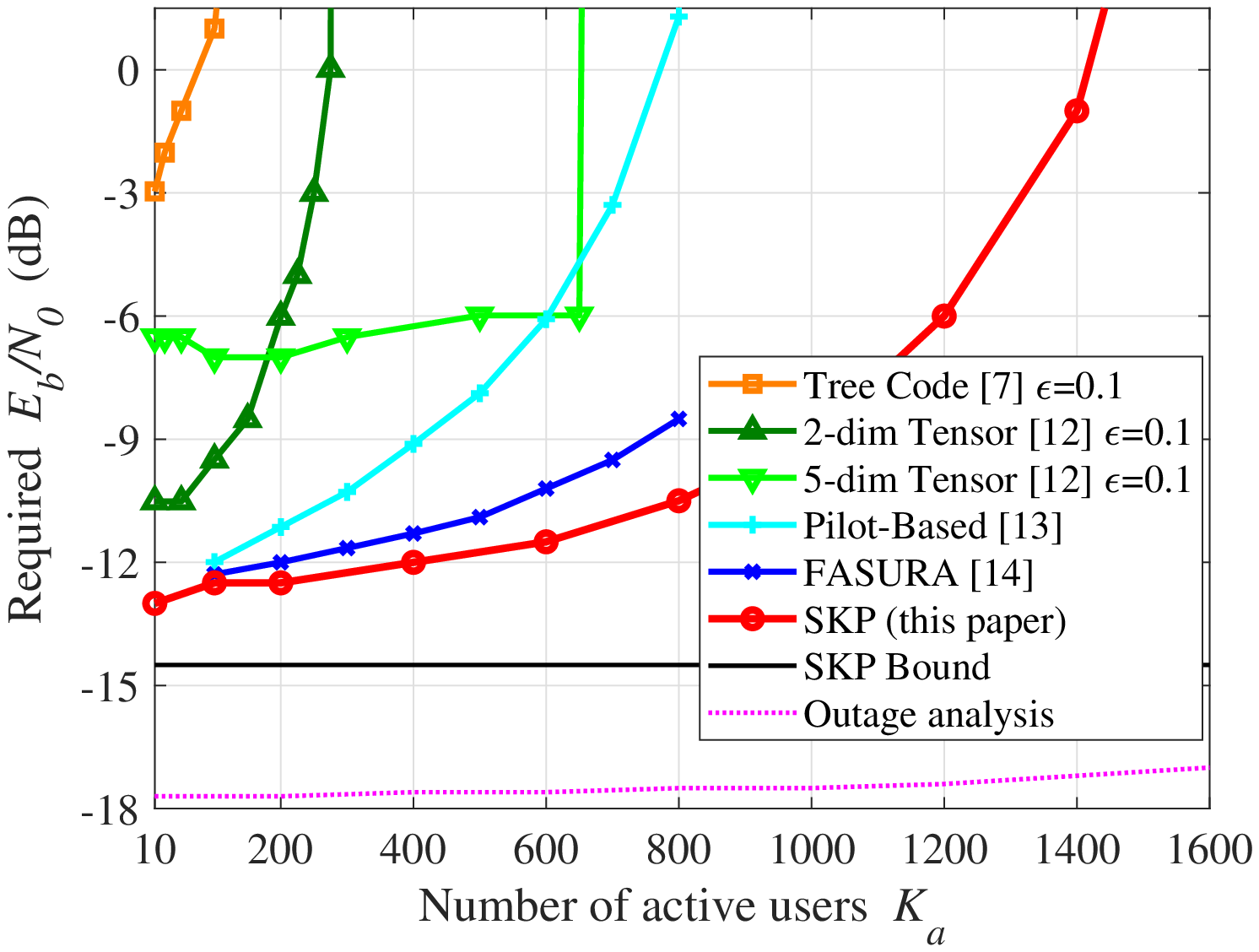}
\end{minipage}
}\hspace{-10pt}
\subfigure[$M=8$]{
\begin{minipage}[c]{0.32\linewidth}
\centering
\includegraphics[width=\linewidth]{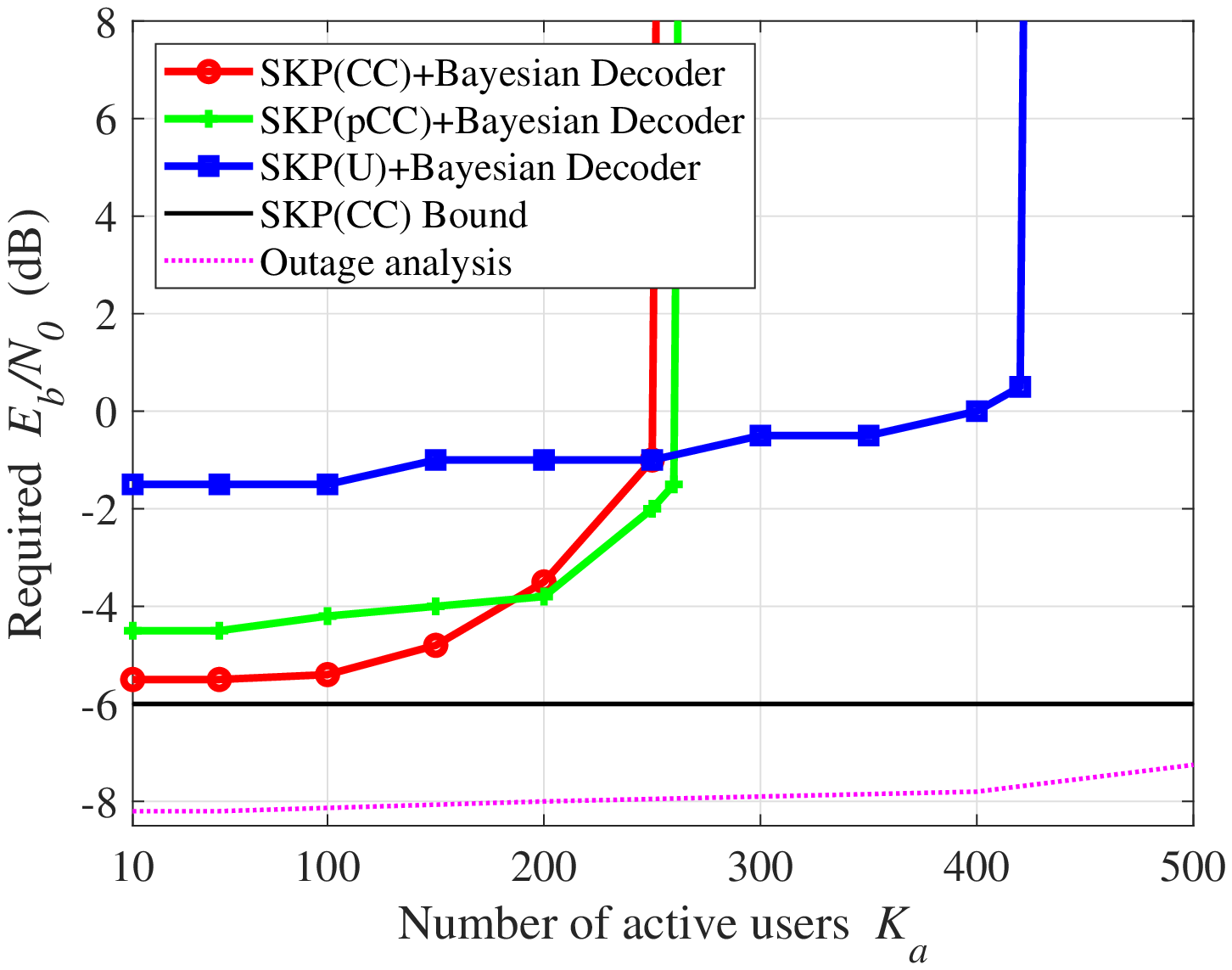}
\end{minipage}
}\hspace{-10pt}
\caption{Required $E_b/N_0$ under various settings: (a) small-size setting; (b) large-size setting; (c) medium-size setting. For ``Tree Code'' in \cite{MIMOURA}, the correlation of channel coefficients between sub-blocks is not considered here.}
\label{result_all}
\vspace{-0.5cm}
\end{figure*}

Fig. \ref{result_all}(a) shows the results \textcolor{black}{of the minimum required $E_b/N_0$} for different schemes with a very small number of antennas, i.e., $M=1$. The SKP bound is derived based on the optimal decoding when $K_a=1$ and $\{\bm{h}_j\}$ are known. It is seen that the proposed scheme has the best performance, and it even outperforms the random coding bound when $K_a\leq 80$ \footnote{For finite block-length, random coding is not optimal; a similar observation has been reported in \cite{sa7_8} for the AWGN channel.}.

Fig. \ref{result_all}(b) shows the results with a large number of antennas, i.e., $M=50$. \textcolor{black}{It is observed that the proposed scheme substantially outperforms the existing MIMO URA schemes in \cite{MIMOURA,TBM,R1,R2}}, and can allow up to 1450 active users.

Fig. \ref{result_all}(c) shows the results with a medium number of antennas, i.e., $M=8$. In this setting, there are no existing results reported in the literature. Thus, we only show the performance of our scheme under three different encoding methods, where ``SKP(CC)'', ``SKP(pCC)'', and ``SKP(U)'' denote the schemes with $R=1/2$ CC, $R=3/4$ punctured CC, and uncoded $\bm{x}_j$, respectively. It is observed that our scheme performs well in this medium-size scenario, and allows up to 420 active users.

\vspace{-0.3cm}
\section{Concluding Remarks}
In this letter, we extended the SKP coding scheme from the AWGN channel to the MIMO channel, and developed an efficient Bayesian receiver. \textcolor{black}{Numerical results show that, the proposed receiver design with the SKP coding performs well in various settings of the MIMO URA scenario.}

\appendices
\vspace{-0.3cm}
\color{black}
\section{Initialization of $\tilde{\mathbb{V}}_a$}

To initialize $\tilde{\mathbb{V}}_a$, we need to calculate the \textit{a posteriori} distribution of each possible $\bm{a}_j$. Define $\bm{g}_j\!=\![\bm{g}_j^{(1)}, \cdots, \bm{g}_j^{(L_a)}]\!=\![a_{j,1}\bm{h}_j^{(1)}, \cdots, a_{j,L_a}\bm{h}_j^{(L_a)}]$ with $\bm{h}_j^{(1)}=\cdots=\bm{h}_j^{(L_a)}=\bm{h}_j$, the \textit{a posteriori} distribution of $\bm{a}_j$ can be written as

\vspace{-0.3cm}
\begin{small}
\begin{equation} \label{dec_prob_appendix1}
\begin{split}
 p\big(\bm{a}_j \big| m_{\bm{Y}\rightarrow{\bm{g}_{j}}}\big) 
& = \int_{\sim\{\bm{a}_{j}\}} p\big(\bm{a}_j, \{\bm{h}_j^{(l_a)}\}_{l_a=1}^{L_a}, \bm{g}_{j} \big| m_{\bm{Y}\rightarrow{\bm{g}_{j}}}\big) \\
& \propto \int_{\sim\{\bm{a}_{j}\}} p\big(m_{\bm{Y}\rightarrow{\bm{g}_{j}}} \big| \bm{g}_{j} \big) p\big(\bm{a}_{j}\big) p\big(\{\bm{h}_j^{(l_a)}\}_{l_a=1}^{L_a}\big) p\big(\bm{g}_{j}\big).
\end{split}
\end{equation}
\end{small}%
Since the exact calculation of (\ref{dec_prob_appendix1}) is complicated, here we ignore the constraint $\bm{h}_j^{(1)}=\cdots=\bm{h}_j^{(L_a)}=\bm{h}_j$ to make each {\small $\int_{\sim\{a_{j,l_a}\}} p(m_{\bm{Y}\rightarrow{\bm{g}_{j}^{(l_a)}}}|\bm{g}_j^{(l_a)}) p(\bm{h}_j) p(\bm{g}_{j}^{(l_a)})\triangleq L(a_{j,l_a}|\bm{g}_{j}^{(l_a)})$} independent and computable separately, as

\vspace{-0.3cm}
\begin{small}
\begin{equation} \label{dec_prob_appendix2}
\begin{split}
 p\big(\bm{a}_j \big| m_{\bm{Y}\rightarrow{\bm{g}_{j}}}\big) 
& \propto p\big(\bm{a}_{j}\big) \int_{\sim\{\bm{a}_{j}\}} p\big(m_{\bm{Y}\rightarrow{\bm{g}_{j}}} \big| \bm{g}_{j} \big) p\big(\{\bm{h}_j^{(l_a)}\}_{l_a=1}^{L_a}\big) p\big(\bm{g}_{j}\big) \\
& \overset{}{\approx} p\big(\bm{a}_{j}\big)\!\prod_{l_a}^{} \int_{\sim\{a_{j,l_a}\}} p\big(m_{\bm{Y}\rightarrow{\bm{g}_{j}^{(l_a)}}}|\bm{g}_j^{(l_a)}\big) p\big(\bm{h}_j^{(l_a)}\big) p\big(\bm{g}_{j}^{(l_a)}\big) \\
& = p\big(\bm{a}_{j}\big)\!\prod_{l_a}^{} L(a_{j,l_a}|\bm{g}_{j}^{(l_a)}).
\end{split}
\end{equation}
\end{small}%

With the simplified model in (\ref{dec_prob_appendix2}), we propose the following two-step method to initialize $\tilde{\mathbb{V}}_a$. We first derive the likelihood of each possible support of $\bm{a}_j$ and select $N_{top}$ likely supports, based on which we then obtain $N_{top}$ choices of $\bm{a}_j$. 

\color{black}
Define $u_{j,i_{IM}}\in\{1,\cdots,L_{IM}\}$ the index of the unique non-zero position in the $i_{IM}$-th segment of $\bm{a}_j, i_{IM}\in\{1,\cdots,I_{IM}\}$. The probability of $u_{j,i_{IM}}$ can be calculated as

\vspace{-0.3cm}
\begin{small}
\begin{equation} \label{ele_si_twice}
\!\Pr\{u_{j,i_{IM}}\!=\!l\}\!\propto\!\sum_{s_i\in\mathcal{S}}{L(a_{j,l}\!=\!s_i|\bm{g}_{j}^{(l)})}\cdot\prod_{l'\neq l} L(a_{j,l'}\!=\!0|\bm{g}_{j}^{(l')}).
\end{equation}
\end{small}%

With (\ref{ele_si_twice}), the probability of each possible support $\bm{u}_j=[u_{j,1}, \cdots, u_{j,I_{IM}}]\in\{1,\cdots,L_{IM}\}^{I_{IM}}$ can be calculated as

\vspace{-0.3cm}
\begin{small}
\begin{equation} \label{all_IM_u}
\begin{split}
\Pr\{\bm{u}_j=[l_1, \cdots, l_{I_{IM}}]\} = \prod_{i_{IM}}{\Pr\{u_{j,i_{IM}}=l_{i_{IM}}\}}.
\end{split}
\end{equation}
\end{small}%
With (\ref{all_IM_u}), the $N_{top}$ most likely $\bm{u}_j$ can be obtained by enumeration, which involves an exponential complexity. To bypass the dilemma, we propose a low-complexity recursive method, as outlined below.

The key point is that a subset $\mathcal{P}$ of $\{\bm{u}_j\}$ with size $O(nI_{IM})$, instead of the full set of $\{\bm{u}_j\}$ of size ${L_{IM}}^{I_{IM}}$, is used at the $n$-th recursion to determine the $n$-th most likely candidate, and is updated for the selection of the later $(n+1)$-th to $N_{top}$-th most likely candidates.

{\textit{Initialization:}} Determine the most likely candidate $\bm{u}_j^{(1)}$ from (\ref{all_IM_u}), and initialize $\mathcal{P}=\{\bm{u}_j^{(1)}\}$.

{\textit{Recursion:}} For each $n, 1\leq n\leq N_{top}$:

\begin{itemize}
  \item Let $\bm{u}_j^{(n)}$ be the candidate with the highest $\Pr\{\bm{u}_j\}$ in $\mathcal{P}$. Delete $\bm{u}_j^{(n)}$ from $\mathcal{P}$.
  \item Add $I_{IM}$ candidates that generated from $\bm{u}_j^{(n)}$ by changing the non-zero position of one of its $I_{IM}$ segments into the next maximum-probability position.
\end{itemize}

\newtheorem{lemma}{Lemma}
\begin{lemma}\label{lemma_01}
The $n$-th ($n\geq 2$) most likely choice can always be generated from one of the top-$(n\!-\!1)$ candidates with only one of its $I_{IM}$ segments changing its non-zero position into the next maximum-probability position.
\end{lemma}

\begin{proof}
All $\Pr\{u_{j,i_{IM}}=l\}$ can be calculated based on (\ref{ele_si_twice}), and ranked in descending order for each $i_{IM}$. Specifically, for each $i_{IM}\in\{1,2,\cdots,I_{IM}\}$, a permutation $\{l_{i_{IM}}^{(1)},l_{i_{IM}}^{(2)},\cdots,l_{i_{IM}}^{(L_{IM})}\}$ of $\{1,2,\cdots,L_{IM}\}$ that satisfies $\Pr\{u_{j,i_{IM}}=l_{i_{IM}}^{(1)}\} \geq\Pr\{u_{j,i_{IM}}=l_{i_{IM}}^{(2)}\} \geq \cdots \geq\Pr\{u_{j,i_{IM}}=l_{i_{IM}}^{(L_{IM})}\}$ can be obtained. Obviously, $\bm{u}_j^{(1)}=\big[l_{1}^{(1)}, l_{2}^{(1)}, \cdots, l_{I_{IM}}^{(1)}\big]$.

Generally, for $\bm{u}_j^{(n)}=\left[l_{1}^{(t_1)}, l_{2}^{(t_2)}, \cdots, l_{I_{IM}}^{(t_{I_{IM}})}\right]$, define the operation $\bm{u}_j^{(n)}\!\downarrow\!i_{IM}=\Big[l_{1}^{(t_1)}, \cdots,$ $ l_{i_{IM}}^{(t_{i_{IM}}+1)}, \cdots, l_{I_{IM}}^{(t_{I_{IM}})}\Big]$, which changes the non-zero position of the $i_{IM}$-th segment into the next maximum-probability position to obtain a worse choice than $\bm{u}_j^{(n)}$.

With this definition, $\bm{u}_j^{(n)}$ can be obtained by operating $\bm{u}_j^{(1)}$ totally $\sum_{i_{IM}=1}^{I_{IM}}{(t_{i_{IM}}-1)}$ times, that is $\bm{u}_j^{(n)}=\bm{u}_j^{(1)}(\!\downarrow\!1)^{t_1-1}(\!\downarrow\!2)^{t_2-1}\cdots(\!\downarrow\!I_{IM})^{t_{I_{IM}}-1}$. Since $\exists i'_{IM}, t_{i'_{IM}}>1$ for $n\geq 2$, we construct a choice $\bm{u}_j^{(n')}\triangleq \bm{u}_j^{(1)}(\!\downarrow\!1)^{t_1-1}\cdots(\!\downarrow\!i'_{IM})^{t_{i'_{IM}}-2}\cdots(\!\downarrow\!I_{IM})^{t_{I_{IM}}-1}$, which is a better choice than $\bm{u}_j^{(n)}$ (thereby $n'<n$) and $\bm{u}_j^{(n)}=\bm{u}_j^{(n')}\!\downarrow\!i'_{IM}$. So far, we construct $n'<n$, and the $\bm{u}_j^{(n)}$ can be generated from $\bm{u}_j^{(n')}$ with only the $i'_{IM}$-th segment changing its non-zero position into the next maximum-probability position. Therefore, lemma \ref{lemma_01} holds.
\end{proof}

\color{black}

Since lemma \ref{lemma_01} holds, the above recursive method selects the $N_{top}$ most likely supports of $\bm{a}_j$.

The complexity of the above recursive algorithm is $\mathcal{O}(I_{IM}N_{top}^2)$. The binary search tree \cite{BST} can be used to keep all possible 
candidates orderly in recursion, which further reduces the complexity of this algorithm to $\mathcal{O}(I_{IM}N_{top}\log{N_{top}})$.

For each support, we obtain an estimate of $\bm{a}_j$ by using the reference symbol in $\bm{a}_j$. Here we obtain each $\bm{a}^{(i_n)}$ by using the rotationally invariant Gaussian mixture (RIGM) method in \cite{RIGM}. Consequently, we obtain initial $\tilde{\mathbb{V}}_a=\{\bm{a}^{(i_n)}\}_{n=1}^{N_{top}}$.

\ifCLASSOPTIONcaptionsoff
  \newpage
\fi

\clearpage


\end{document}